\documentclass[preprint]{sigplanconf}

\usepackage[T1]{fontenc}
\usepackage[safe]{tipa} 
\usepackage[utf8]{inputenc}
\usepackage[english]{babel}
\usepackage{etex}
\usepackage{stmaryrd}
\usepackage{amsmath}
\usepackage{amsfonts}
\usepackage{amssymb}
\usepackage{mathtools}
\usepackage[amsmath,thmmarks]{ntheorem}
\usepackage{tikz}
\usepackage{wrapfig}
\usepackage[numbers]{natbib}
\usepackage{mathpazo}
\usepackage{microtype}
\usepackage{booktabs}
\usepackage{mathpartir}
\usepackage{upgreek}
\usepackage{enumerate}

\newcommand\pfun{\mathrel{\ooalign{\hfil$\mapstochar\mkern5mu$\hfil\cr$\to$\cr}}}

\usepackage{hyperref}

\usetikzlibrary{shapes.geometric}
\usetikzlibrary{calc}
\usetikzlibrary{shapes.multipart}
\usetikzlibrary{arrows}

\makeatletter
\let\UrlSpecialsOld\UrlSpecials
\def\UrlSpecials{\UrlSpecialsOld\do\/{\Url@slash}\do\_{\Url@underscore}}%
\def\Url@slash{\@ifnextchar/{\kern-.11em\mathchar47\kern-.2em}%
   {\kern-.0em\mathchar47\kern-.08em\penalty\UrlBigBreakPenalty}}
\def\Url@underscore{\nfss@text{\leavevmode \kern.06em\vbox{\hrule\@width.3em}}}
\makeatother

\theorembodyfont{}
\newtheorem{theorem}{Theorem}
\newtheorem{lemma}[theorem]{Lemma}
\theoremstyle{nonumberplain}
\theoremheaderfont{\scshape}
\theoremsymbol{\ensuremath{\blacksquare}}
\theoremseparator{.}
\newtheorem{proof}{Proof}

\usepackage{listings}
\newcommand{\li}{\lstinline[style=Haskell]}
\lstnewenvironment{haskell}{\lstset{style=Haskell}}{}
\lstdefinestyle{Haskell}{language=Haskell
        ,columns=flexible
	,basewidth={.365em}
	,keepspaces=True
        ,texcl=true
        ,basicstyle=\sffamily
        ,stringstyle=\itshape
        ,showstringspaces=false
        ,literate={->}{$\,\to\,$}2
                  {<-}{$\,\leftarrow\,$}2
                  {=>}{$\,\Rightarrow\,$}2
                  {→}{$\,\to\,$}2
                  {>>}{{>>}\hspace{-1pt}}2
                  {[]}{[\,]\ }1
                  {++}{{$+\!\!+$\ }}1
                  {\ .\ }{{$\,\circ\,$}}2
	,keywords={type,data,where,let,in,case,of}
        }

\newcommand{\ci}{\lstinline[style=Cmm]}
\lstnewenvironment{cmm}{\lstset{style=Cmm}}{}
\lstdefinestyle{Cmm}{language=C
        ,columns=fullflexible
        ,texcl=true
	,commentstyle=\sffamily\itshape
        ,basicstyle=\small\ttfamily
        }

\conferenceinfo{Haskell Implementors Workshop}{Sep 14 2012, Copenhagen} 
\copyrightyear{2012} 
\copyrightdata{Joachim Breitner}

\titlebanner{---preprint---}
\preprintfooter{Joachim Breitner: dup -- Explicit un-sharing in Haskell}
\title{dup -- Explicit un-sharing in Haskell}

\authorinfo{Joachim Breitner}{Karlsruhe Institute of Technology}{breitner@kit.edu}

\begin{document}
\maketitle

\begin{abstract}
We propose two operations to prevent sharing in Haskell that do not require modifying the data generating code, demonstrate their use and usefulness, and compare them to other approaches to preventing sharing. Our claims are supported by a formal semantics and a prototype implementation.
\end{abstract}




\section{Introduction}

Thanks to the immutable nature of data in a pure functional programming language such as Haskell, there are many possibilities for sharing, i.e.\ one object in memory can used in multiple places in the program. In general, this is a good thing, as it can save both execution time (by not calculating the data again) and memory space (by not copying the data).

But there are cases where sharing can hurt, and sometimes hurt badly. A famous example (\citep{wikibook,lazy-evaluators}) is the following function:
\begin{haskell}
let l = [1..100000000]
    f :: [Int] -> Int
    f xs = last xs + head xs
in  f l
\end{haskell}
This program is space-leaky and will quickly run out of memory. If we substitute the term for \li-xs- in the body of \li-f- and evaluate that expression, it runs quickly and in constant memory. We have avoided the sharing of \li-xs- between the calls to \li-last- and \li-head- and the list elements can be garbage collected as soon as they have been consumed by \li-last-. This came at the expense of evaluating the list twice, which is fine, as the list is \emph{large} but \emph{cheap to calculate}.

But this source transformation, as well as other source transformations to avoid sharing (see Section \ref{sec:unit} and \ref{sec:church}), is not always possible or desirable, e.g.\  when the parameter passed to \li-f- comes from library code not under the control of the programmer. Therefore, we propose a new primitive operation \li-dup- which copies a (possibly unevaluated) value on the heap.
\begin{haskell}
data Box a = Box a
dup :: a -> Box a
\end{haskell}
Its value semantics are that of \li!(\x -> Box x)!; the wrapping in \li-Box- just serves the purpose of controlling the exact point of execution of \li-dup- by case-analyzing the \li-Box-. Using \li-dup- allows us to modify in the above example only the code of \li-f- to prevent sharing and achieve constant memory usage:
\begin{haskell}
let f xs = case dup xs of
    	Box xs' -> last xs' + head xs
    l = [1..100000000]
in  f l
\end{haskell}
In Section \ref{sec:unsharing}, we demonstrate the use of \li-dup- and other approaches on the more elaborate example introduced in Section \ref{sec:example}, taking on the programmer’s point of view.

An sharp-witted reader with knowledge of a typical implementation of a Haskell runtime might already have noticed that just copying the object on the heap representing the parameter \li-xs- might not be enough: If, for example, the first cons-cell of \li-xs- is already evaluated, then \li-dup xs- will copy that cell, but the thunk representing the tail of the list will still be shared between \li-xs'- and \li-xs-, and \li-f- will again devour memory. Such things may occur without the programmer’s knowledge, e.g.\ during a compiler optimization pass.

To that end, we propose a variant of \li-dup-, called \li-deepDup-, which effectively copies the complete heap referenced by its argument. This happens -- as one would expect for anything related to Haskell -- lazily: The objects referenced  by the parameter are copied if and when they are needed. In other words: After having evaluated a function which only works on \li-deepDup-’ed copies of its parameters, nothing this evaluation has created on the heap is referenced anymore, unless it is referenced by the function's return value (this is formalized in Theorem \ref{thm:deepdup}).

Our specific contributions are:
\begin{itemize}
\item We introduce primitives that give the programmer the possibility to explicitly prevent sharing.
\item In contrast to approaches based on source transformations, using \li-dup- and \li-deepDup- does \emph{not} require changes to the generating code.
\item We provide precise semantics in the context of Launchbury’s natural semantics for Lazy Evaluation (Section \ref{sec:semantics}) and prove that the recursive variant \li-deepDup- is effective.
\item We show the feasibility of our approach using a proof-of-concept implementation targeting code compiled by an unmodified GHC. (Section \ref{sec:prototype})
\end{itemize}

\section{The running example}
\label{sec:example}

For the remainder of the paper, we will use one running example to demonstrate and discuss the use of \li-dup-. The task at hand, inspired by the minimax algorithm that searches for an optimal strategy in a two-player turn-based game, is to find a path through a (possibly infinite) tree that maximizes some valuation of the nodes. So abstractly, we have a type \li-S- of states, a valuation function \li-value-, an initial state \li-init- and for every state \li-s-, a list of successor states \li-succs s-. For the sake of simplicity of the presentation, we assume this \li-succs s- to be always non-empty (see also Figure~\ref{fig:ex}).

\begin{figure}
\begin{haskell}
-- The problem specification
type S = ...
init :: S
succs :: S -> [S]
value :: S -> Integer

-- The search tree code
data Tree = Node S [Tree]
fstChild :: Tree -> Tree
fstChild (Tree _ (x:xs)) = x

tree :: S -> Tree
tree s = Node s (map tree (succs s))

solve :: Tree -> [S]
solve (Node n ts) = n : solve picked
  where
  rated = [ (t, rate depth t) | t <- ts ]
  picked = fst (maximumBy (comparing snd) rated)
  depth = ...

rate :: Int -> Tree -> Integer
rate 0 (Node s _) = value s
rate d (Node _ ts) = maximum (map (rate (d-1)) ts)

main = do
  let t = tree init
  print $ solve t !! 10000
  doSomethingElseWith t
\end{haskell}
\caption{The running example}
\label{fig:ex}
\end{figure}

Based on these functions, we define a search tree and a solver. The solver picks the successor with the highest rating, whereas the rating is the highest value of nodes at a configurable depth.

Assume a constant number of successors $b$, $b>0$, and that the value of \li-depth- is $d$. Consider what happens when we want to calculate the first 10\,000 elements of the solution: The \li-rate- function will evaluate lots of nodes that will \emph{not} be picked for the solution. But as they are still referenced by the tree \li-t-, the garbage collector cannot get rid of them. So in addition to the  10\,000 interesting nodes, roughly $10\,000\cdot (b-1)\cdot b^{d-1}$ nodes are evaluated that the programmer knows are not required to be kept around. The first row of Figure \ref{fig:evaluation} depicts the heap during this evaluation, with $d=1$ and $b=2$.

\newcommand{\stats}[1]{\ref*{#1}}
More concretely with $d=4$, $b=4$, \li-type S = Word32- and a very cheap \li-succs- and \li-value- functions, this program requires \stats{stats:Original:Shared:mem}~MB of system memory (as reported by the GHC runtime as “total memory in use” when passed the \ci!-s! option) and runs in \stats{stats:Original:Shared:time} seconds.%
\footnote{All statistics are obtained on a machine with 2~GHz and sufficient (32~GB) RAM. The complete code used to generate these statistics is available in the \ci!ghc-dup! repository at \url{http://darcs.nomeata.de/ghc-dup}.} Sharing is indeed the problem here: If we remove the last line of \li-main-, the program runs in \stats{stats:Original:Unshared:mem}~MB of memory and takes \stats{stats:Original:Unshared:time} seconds.

\section{Unsharing the example}
\label{sec:unsharing}

\begin{figure*}
\centering
\input{statstable}
\caption{Time and space performance for $b=4$ and $d=4$}
\label{fig:stats}
\end{figure*}

\newcommand{\thefirstevaluationmatrix}{
\node[right] {\rmfamily original:};
\&
\node[solve] (root) {T};
\draw (root.north) -- +(0mm,2mm);
\&

\&

\node[solve] (root) {N}
child {node [rate] {T}}
child {node  {T}};
\draw (root.north) -- +(0mm,2mm);

\&
\node[solve] (root) {N}
child {node [rate] {N}
	child {node  {N} child {node {T}} child {node {T}}}
	child {node  {N} child {node {T}} child {node {T}}}
	}
child {node  {T}};
\draw (root.north) -- +(0mm,2mm);

\&
\node[solve] (root) {N}
child {node  {N}
	child {node  {N} child {node {T}} child {node {T}}}
	child {node  {N} child {node {T}} child {node {T}}}
	}
child {node[rate]  {N}
	child {node  {N} child {node {T}} child {node {T}}}
	child {node  {N} child {node {T}} child {node {T}}}
	}
;
\draw (root.north) -- +(0mm,2mm);

\&

\node (root) {N}
child {node  {N}
	child {node  {N} child {node {T}} child {node {T}}}
	child {node  {N} child {node {T}} child {node {T}}}
	}
child {node [solve] {N}
	child {node  {N} child {node {T}} child {node {T}}}
	child {node  {N} child {node {T}} child {node {T}}}
	};
\draw (root.north) -- +(0mm,2mm);

\\
\node[right]{\rmfamily \textsf{solveDup}:};
\&
\node[solve] (root) {T};
\draw (root.north) -- +(0mm,2mm);
\&

\node  at (-6mm,0mm) (root) {T};
\draw (root.north) -- +(0mm,2mm);
\node[solve] (root2) {T};
\draw[double] (root) -- (root2);
\&

\node  at (-6mm,0mm) (root) {T};
\draw (root.north) -- +(0,2mm);
\node[solve](root2) {N}
child {node [rate] {T}}
child {node  {T}};
\draw[double] (root) -- (root2);

\&
\node at (-6mm,0mm)  (root) {T};
\draw (root.north) -- +(0,2mm);
\node[solve](root2) {N}
child {node [rate] {N}
	child {node  {N} child {node {T}} child {node {T}}}
	child {node  {N} child {node {T}} child {node {T}}}
	}
child {node  {T}};
\draw (root.north) -- +(0,2mm);
\draw[double] (root) -- (root2);

\&
\node at (-6mm,0mm)  (root) {T};
\draw (root.north) -- +(0,2mm);
\node[solve](root2) {N}
child {node {N}
	child {node  {N} child {node {T}} child {node {T}}}
	child {node  {N} child {node {T}} child {node {T}}}
	}
child {node [rate] {N}
	child {node  {N} child {node {T}} child {node {T}}}
	child {node  {N} child {node {T}} child {node {T}}}
	}
;
\draw (root.north) -- +(0,2mm);
\draw[double] (root) -- (root2);

\&

\node  at (-6mm,0mm) (root) {T};
\draw (root.north) -- +(0,2mm);
\draw [gray] node(root2) {N}
child {node  {N}
	child {node  {N} child {node {T}} child {node {T}}}
	child {node  {N} child {node {T}} child {node {T}}}
	}
child {node[solve,black] {N}
	child[black] {node  {N} child {node {T}} child {node {T}}}
	child[black] {node  {N} child {node {T}} child {node {T}}}
	};
\draw (root.north) -- +(0,2mm);
\draw[gray,double] (root) -- (root2);

\\
\node[right]{\rmfamily \textsf{rateDup}:};
\&
\node[solve] (root) {T};
\draw (root.north) -- +(0mm,2mm);
\&

\node[solve] (root) {N}
child {node [rate] {T}}
child {node  {T}};
\draw (root.north) -- +(0mm,2mm);

\&
\node[solve] (root) {N}
child {node (rate) {T}}
child {node  {T}};
\path (rate) +(6mm,0) node[rate] (rate2) {T};
\draw (root.north) -- +(0mm,2mm);
\draw[double] (rate) -- (rate2);

\&
\node[solve] (root) {N}
child {node (rate) {T}
 +(6mm,0) node[rate] (rate2) {N}
	child {node  {N} child {node {T}} child {node {T}}}
	child {node  {N} child {node {T}} child {node {T}}}
	}
child {node  {T}};
\draw (root.north) -- +(0mm,2mm);
\draw[double] (rate) -- (rate2);

\&
\node[solve] (root) {N}
child {node (rate') {T}
 +(6mm,0) node[gray] (rate'2) {N}
	child[gray] {node  {N} child {node {T}} child {node {T}}}
	child[gray] {node  {N} child {node {T}} child {node {T}}}
	}
child {node (rate) {T}
 +(6mm,0) node[rate] (rate2) {N}
	child {node  {N} child {node {T}} child {node {T}}}
	child {node  {N} child {node {T}} child {node {T}}}
	}
;
\draw (root.north) -- +(0mm,2mm);
\draw[double] (rate) -- (rate2);
\draw[double,gray] (rate') -- (rate'2);

\&

\node (root) {N}
child {node (t1) {T}
 +(6mm,0) node[gray] (t1') {N}
	child[gray] {node  {N} child {node {T}} child {node {T}}}
	child[gray] {node  {N} child {node {T}} child {node {T}}}
	}
child {node [solve] (t2) {T}
 +(6mm,0) node[gray] (t2') {N}
	child[gray] {node  {N} child {node {T}} child {node {T}}}
	child[gray] {node  {N} child {node {T}} child {node {T}}}
	};
\draw (root.north) -- +(0mm,2mm);
\draw[gray,double] (t1) -- (t1');
\draw[gray,double] (t2) -- (t2');
\\
}

\newcommand{\thesecondevaluationmatrix}{
\node[right]{\rmfamily \textsf{solveDup}:};
\&
\node[solve] (root) {N}
child {node  {T}}
child {node  {T}};
\draw (root.north) -- +(0mm,2mm);
\&

\node  at (-6mm,0mm) (root) {N};
\draw (root.north) -- +(0mm,2mm);
\node[solve] (root2) {N}
child {node  {T}}
child {node  {T}};
\draw[double] (root) -- (root2);
\draw (root) -- (root2-1);
\draw (root) -- (root2-2);
\&

\node  at (-6mm,0mm) (root) {N};
\draw (root.north) -- +(0,2mm);
\node[solve](root2) {N}
child {node [rate] {T}}
child {node  {T}};
\draw[double] (root) -- (root2);
\draw (root) -- (root2-1);
\draw (root) -- (root2-2);

\&
\node at (-6mm,0mm)  (root) {N};
\draw (root.north) -- +(0,2mm);
\node[solve](root2) {N}
child {node [rate] {N}
	child {node  {N} child {node {T}} child {node {T}}}
	child {node  {N} child {node {T}} child {node {T}}}
	}
child {node  {T}};
\draw (root.north) -- +(0,2mm);
\draw[double] (root) -- (root2);
\draw (root) -- (root2-1);
\draw (root) -- (root2-2);

\&

\node  at (-6mm,0mm) (root) {N};
\draw (root.north) -- +(0,2mm);
\draw  node[solve] (root2) {N}
child {node  {N}
	child {node  {N} child {node {T}} child {node {T}}}
	child {node  {N} child {node {T}} child {node {T}}}
	} 
child {node[rate] {N}
	child {node  {N} child {node {T}} child {node {T}}}
	child {node  {N} child {node {T}} child {node {T}}}
	};
\draw (root.north) -- +(0,2mm);
\draw[double] (root) -- (root2);
\draw (root) -- (root2-1);
\draw (root) -- (root2-2);

\&

\node  at (-6mm,0mm) (root) {N};
\draw (root.north) -- +(0,2mm);
\draw  node[gray] (root2) {N}
child {node  {N}
	child[black] {node  {N} child {node {T}} child {node {T}}}
	child[black] {node  {N} child {node {T}} child {node {T}}}
	edge from parent[gray]
	} 
child {node[solve,black] {N}
	child[black] {node  {N} child {node {T}} child {node {T}}}
	child[black] {node  {N} child {node {T}} child {node {T}}}
	edge from parent[gray]
	};
\draw (root.north) -- +(0,2mm);
\draw[gray,double] (root) -- (root2);
\draw (root) -- (root2-1);
\draw (root) -- (root2-2);

\\
\node[right]{\rmfamily \textsf{solveDeepDup}:};
\&
\node[solve] (root) {N}
child {node  {T}}
child {node  {T}};
\draw (root.north) -- +(0mm,2mm);
\&

\node  at (-6mm,0mm) (root) {N}
child {node  {T}}
child {node  {T}}
;
\draw (root.north) -- +(0mm,2mm);
\node[solve] (root2) {N}
child {node[rate]  {D}}
child {node  {D}};
\draw[double] (root) -- (root2);
\draw[->] (root2-1) -- (root-1);
\draw[->] (root2-2) -- (root-2);

\&

\node  at (-6mm,0mm) (root) {N}
child {node  {T}}
child {node  {T}}
;
\draw (root.north) -- +(0mm,2mm);
\node[solve] (root2) {N}
child {node[rate]  {T}}
child {node  {D}};
\draw[double] (root) -- (root2);
\draw[double] (root2-1) -- (root-1);
\draw[->] (root2-2) -- (root-2);

\&

\node  at (-6mm,0mm) (root) {N}
child {node  {T}}
child {node  {T}}
;
\draw (root.north) -- +(0mm,2mm);
\node[solve] (root2) {N}
child {node[rate]  {N}
	child {node  {N} child {node {T}} child {node {T}}}
	child {node  {N} child {node {T}} child {node {T}}}
	}
child {node  {D}};
\draw[double] (root) -- (root2);
\draw[double] (root2-1) -- (root-1);
\draw[->] (root2-2) -- (root-2);

\&
\node  at (-6mm,0mm) (root) {N}
child {node  {T}}
child {node  {T}}
;
\draw (root.north) -- +(0mm,2mm);
\node[solve] (root2) {N}
child {node  {N}
	child {node  {N} child {node {T}} child {node {T}}}
	child {node  {N} child {node {T}} child {node {T}}}
	}
child {node[rate]  {N}
	child {
		node  {N} child {node {T}} child {node {T}}
		}
	child {
		node  {N} child {node {T}} child {node {T}}
		}
	edge from parent
};
\draw[double] (root) -- (root2);
\draw[double] (root2-1) -- (root-1);
\draw[double] (root2-2) -- (root-2);

\&
\node  at (-6mm,0mm) (root) {N}
child {node  {T}}
child {node  {T}}
;
\draw (root.north) -- +(0mm,2mm);
\node[gray] (root2) {N}
child[gray] {node  {N}
	child {node  {N} child {node {T}} child {node {T}}}
	child {node  {N} child {node {T}} child {node {T}}}
	}
child {node[solve]  {N}
	child[black] {
		node  {N} child {node {T}} child {node {T}}
		}
	child[black] {
		node  {N} child {node {T}} child {node {T}}
		}
	edge from parent[gray]
};
\draw[double,gray] (root) -- (root2);
\draw[double,gray] (root2-1) -- (root-1);
\draw[double] (root2-2) -- (root-2);

\\
}

\begin{figure*}
\centering
{
\sffamily
\begin{tikzpicture}
[level/.style={sibling distance=20mm/2^#1,level distance=6mm},
every node/.style={inner sep=1pt},
solve/.style={draw,circle,inner sep=1pt},
rate/.style={draw,rectangle,inner sep=2pt},
]

\matrix[column sep={4mm},row sep=4mm, ampersand replacement=\&] (matrix) {
\thefirstevaluationmatrix
\begin{scope}
\path[use as bounding box] (0pt,-5em) -- (0pt,0em);
\node[above right] (caption) {};
\end{scope}
\\
\thesecondevaluationmatrix
};

\path (caption.north west)
node [above right] {
\parbox{88mm}{
\rmfamily
\textsf{T}: thunk, \textsf{N}: node, $=$: dup’ed closure,
 \textcolor{gray}{garbage},
\\
\begin{tikzpicture}[baseline=(n.base)]
\node[solve] (n) {T};
\end{tikzpicture}:
 current argument of solve,
\begin{tikzpicture}[baseline=(n.base)]
\node[rate] (n) {T};
\end{tikzpicture}:
 current argument of rate
\\
}
};

\path (matrix.south west)
node [above right] {
\parbox{88mm}{
\rmfamily
\textsf{D}: \textsf{deepDup} application thunk
}
};
\end{tikzpicture}}\\
\raisebox{23em}[0pt][0pt]{%
\parbox{\linewidth}{
\caption{The heap during original and \textsf{dup}’ed evaluation with $b=2$ and $d=1$}
\label{fig:evaluation}
}}\\
\caption{Comparing \textsf{solveDup} and \textsf{solveDeepDup} applied to a partly evaluated tree with $b=2$ and $d=1$}
\label{fig:deepevaluation}
\end{figure*}

We want to improve the space performance of the program in the example and thus, due to the saved work in the garbage collector, also the runtime performance. In the following, we use \li-dup-, first wrapping the argument of \li-solve-, then the argument of \li-rate-, and \li-deepDup-. We also try two variants that work without new primitives, but require refactoring the generating code. The statistics are collected in Figure \ref{fig:stats}, where all six strategies are applied to
\begin{itemize}
\item an otherwise unreferenced tree, i.e.\ the example code without the last line of the main function,
\item a shared, unevaluated tree as shown in Figure \ref{fig:ex},
\item a shared, unevaluated tree wrapped in another thunk, by passing \li-(fstChild t)- to \li-solve-,
\item a shared tree that has been partly evaluated forcing \li-seq (fstChild t)- before passing \li-t- to the solver,
\item a shared tree that has been fully evaluated by the unmodified solver before,
\item a shared, unevaluated tree that is processed twice by the (possibly modified) solver.
\end{itemize}
In the two variants based on refactoring, the data type used for the tree does not allow for partial or full evaluation, so these runs are omitted.

\subsection{Using dup}

We now modify the example to use our new primitives. There are a few choices in doing so, with different trade-offs. One candidate for \li-dup-’ing is the function \li-solve-: We know that the parameter \li-t- to \li-solve- is an unevaluated expression, and decoupling that from the \li-t- that we pass to \li-doSomethingElseWith- will allow the garbage collector to clean up the tree as \li-solve- proceeds to process it (Figure \ref{fig:evaluation}, second row). So we wrap \li-solve- in \li-solveDup- and use that in \li-main-.
\begin{haskell}
solveDup t = case dup t of Box t' -> solve t'
\end{haskell}
And indeed, we have almost achieved the performance of the original program without sharing: \stats{stats:SolveDup:Shared:mem}~MB and \stats{stats:SolveDup:Shared:time} seconds.

Another candidate for \li-dup-’ing is the function \li-rate-:
As this is the function whose return value is taken into account when deciding whether to pick the argument or not, we know that in most cases, its argument will not be used any more. Therefore, by creating a wrapper \li-rateDup- that duplicates the argument, and using that in \li-solve-, we allow for the argument and all its children to be garbage collected once \li-rate- has finished.
\begin{haskell}
rateDup d t = case dup t of Box t' -> rate d t'
\end{haskell}

Both the runtime and the memory footprint of the program are greatly reduced compared to the original program: It uses \stats{stats:RateDup:Shared:mem}~MB of memory and takes \stats{stats:RateDup:Shared:time} seconds to finish. It is surprising that this even surpasses the speed of the original program without sharing. The reason is that with \li-rate- wrapped in \li-dup-, the first child of the node under inspection of \li-solve- can be freed already when its next child is evaluated by \li-rate- (Figure~\ref{fig:evaluation}, last row, second-to-last column), so the copying garbage collector needs to do even less work.

\begin{figure*}
\centering
\input{statstable-slow}
\caption{Time and space performance for $b=4$ and $d=4$ using an expensive \li-succs- function.}
\label{fig:statsexpensive}
\end{figure*}

\subsection{Using deepDup}
\label{sec:deepdup}
Using \li-dup- is a fragile business and requires the programmer to have a very good idea about what is happening at runtime. It will fail, for example, in two common situations: If the call to \li-solveDup- in \li-main- in Figure \ref{fig:ex} would not just pass the tree \li-t- but rather an expression referencing \li-t-, e.g.
\begin{haskell}
print $ solveDup (fstChild t) !! 1000
\end{haskell}
then \li-dup- will only copy this unevaluated expression, but both copies will reference the same unevaluated expression for~\li-t-, and we are back at the original performance (\stats{stats:SolveDup:SharedThunk:mem}~MB, \stats{stats:SolveDup:SharedThunk:time} seconds).

The same effect occurs if the tree is already partly evaluated. This may even be caused by a compiler transformation, e.g.\ the wrapper/worker transformation, assuming that \li-doSomethingElseWith- is strict  in its argument \citep{unboxed}. Then, the parameter \li-t- is the \li-Node- constructor referencing other nodes or unevaluated trees, and copying the constructor does not help to prevent sharing the referenced data, as shown in the first row of Figure \ref{fig:deepevaluation}.

This is where \li-deepDup- comes in:
Intuitively, \li-deepDup- takes a complete and private copy of the entire heap reachable from its argument, hence preventing any unwanted evaluation outside this copy. In fact this is done lazily: It will just copy the object specified by its parameter, and change all references therein so that before they are evaluated, \li-deepDup- copies them.

So by wrapping \li-solve- in a call to \li-deepDup-:
\begin{haskell}
solveDeepDup t = case deepDup t of Box t' -> solve t'
\end{haskell}
we achieve the performance of a successful run with \li-dup- (\stats{stats:SolveDeepDup:SharedThunk:mem}~MB and \stats{stats:SolveDeepDup:SharedThunk:time} seconds), but also in the cases where \li-t- has already been partly evaluated or is wrapped in another unevaluated expression. The second row of Figure \ref{fig:deepevaluation} shows \li-deepDup- at work.

Using \li-deepDup- is therefore more reliable and easier to handle: The programmer need not have an exact idea of the evaluation state of the arguments when \li-deepDup- is called. And the recursive copying is surprisingly cheap: Even when the tree is already fully evaluated, e.g. by an earlier call to \li-solve t !! 10000-, the runtime stays the same within the precision of the benchmark.


\subsection{The unit type argument pattern}

The problem at hand is, of course, not new, and Haskell programmers have solved it one way or the other before, by rewriting the code to allow more control over sharing.

\label{sec:unit}

A common approach is to replace values that you do not want to be shared by functions, e.g.\ by turning a bound expression \li-let x = e- into a lambda expression \li!let x = \() -> e!. At every point in the program where \li-e- is required, one can get the value of it using \li-x ()-; there will be no sharing between different calls to \li-x ()-

One needs to be careful, though, as some compiler optimizations can introduce unwanted sharing again. The code
\begin{haskell}
xs :: () -> [Int]
xs () = [1..10000000]

main = do
    print (last (xs ()))
    print (head (xs ()))
\end{haskell}
works as expected without optimization. Passing  \ci!-O! to GHC results in sharing again, as a result of the full laziness transformation. In fact, in a discussion of this example on the GHC bug tracker \citep{spaceleakbug}, Claus Reinke suggests an operation like \li-dup- to solve this.%
\footnote{
If, however, the type signature of \li-xs- is not given, then no unwanted sharing happens even with \ci!-O!. The inferred most general type of \li-xs- is polymorphic with type class constraints. This implies that additional parameters are being passed under the hood and they successfully prevent sharing.
}

Applying this pattern to our problem, and aiming for a tree with unshareable subtrees, we can define the following types:
\begin{haskell}
data UTree' = UNode S [UTree]
type UTree = () -> UTree'
\end{haskell}
The required changes to the functions on trees are mechanical and guided by the type checker. The resulting code, when not hit by some optimization-induced re-sharing, shows very good time and space complexity. If sharing is desired at some points of the program, those parts will have to work with the regular \li-Tree- type, possibly leading to a duplication of code.

\subsection{Church encoding}
\label{sec:church}

An alternative is to restructure the program so that the value that must not be shared is not represented using data constructors but rather as a higher-order function \citep{churchenc,olegchurchenc}. This transformation is known as the Church encoding of a data type, or a variant thereof. For the algebraic tree data type in our running example, we would obtain the following type and conversion functions:
\begin{haskell}
type CTree = forall a. (S -> [a] -> a) -> a

toCTree :: Tree -> CTree
toCTree (Node s ts) f = f s $ map (\t -> toCTree t f) ts

fromCTree :: CTree -> Tree
fromCTree ct = ct Node
\end{haskell}

A church-encoded tree corresponding to the value \li-tree s- can be nicely created with the following code:
\begin{haskell}
ctree :: S -> CTree
ctree s f = f s $ map (\s' -> ctree s' f) (succs s)
\end{haskell}

Unfortunately, adapting \li-solve- to this type is a non-trivial task, as the two recursions happening therein (\li-solve- and \li-rate-) need to be folded into one pass:
\begin{haskell}
csolve :: CTree -> [S]
csolve t = fst (t csolve')
  where
  csolve' :: S -> [([S], Int -> Int)] -> ([S], Int -> Int)
  csolve' n rc = 
    ( n : fst (maximumBy (comparing (($ depth) . snd)) rc)
    , \d -> if d == 0 then value n
                      else maximum (map (($ d-1) . snd) rc))
    
\end{haskell}
This additional complexity might make this approach impractical in larger settings.
Note, though, that applying this pattern to the list data type turns a list into its right fold and can enable deforestation \citep{deforestation}.

\pagebreak[3]
\subsection{Comparison and interpretation}

As we can see from the statistics in Figure \ref{fig:stats}, the unit type argument pattern is the clear winner in both runtime and space performance. It is ahead of \li-rateDup- for the same reason that made \li-rateDup- faster than \li-solveDup-: Now even the subtrees in recursive calls of \li-rate- are freed immediately.
Unfortunately, it requires a thorough refactoring of both the data generating and data consuming code; all combinators working on the data type need to be carefully rewritten to preserve the non-sharing behavior of the lifted data type. Also, the full laziness transformation can break the pattern, making it slightly fragile.

The church encoding pattern shows good and predictable memory performance, but exhibits slightly worse runtime behavior. The cases where it is ahead of other approaches it wins only due to the garbage collector overhead induced by unprevented sharing. As the previous pattern, it requires extensive refactoring.

Our primitives come with very small overhead when applied to data that is actually unshared, as we show in the first column. In fact, careful use of \li-dup- can improve performance noticeably even if only small pieces of data can be un-shared and thus freed quickly. While \li-dup- is subtle to use, \li-deepDup- is robust and its effect is more precisely defined, as shown in the next section.

Obviously, avoiding sharing is a bad idea when the result is expensive to create. Figure~\ref{fig:statsexpensive} runs the same benchmark with an expensive (but otherwise equal) \li-succs- function. If the tree needs to be processed twice, then throwing the result away after the first run (as done in the last column) results in a serious loss of run-time performance. Also for the same reason that the \li-rateDup- and unit lifting variants were faster before they now slow down the program, as parts of the tree are evaluated twice.

On the other hand, if the memory footprint becomes larger than the available memory, being able to run the program slowly is still better than not being able to run it at all, so even in this case there can be uses for \li-dup- and \li-deepDup-.

\pagebreak[3]
\section{A natural semantics}
\label{sec:semantics}

To substantiate our claims about the usefulness of \li-dup- and especially \li-deepDup-, we give them a precise meaning within Launchbury’s natural semantics for lazy evaluation \citep{launchbury} and prove that all memory allocated by a function whose arguments are wrapped with \li-deepDup- can be freed after the function has been completely evaluated.

We extend Launchbury’s semantics for normalized lambda calculus with our two primitives:
\newcommand{\mdup}{\text{\textsf{dup}}}
\newcommand{\mdeepDup}{\text{\textsf{deepDup}}}
\newcommand{\sVar}{\text{Var}}
\newcommand{\sExp}{\text{Exp}}
\newcommand{\sHeap}{\text{Heap}}
\newcommand{\sVal}{\text{Val}}
\newcommand{\sValue}{\text{Value}}
\newcommand{\sEnv}{\text{Env}}
\newcommand{\sApp}[2]{\operatorname{#1}#2}
\newcommand{\sLam}[2]{\text{\textlambda} #1.\, #2}
\newcommand{\sDup}[1]{\sApp \mdup #1}
\newcommand{\sDeepDup}[1]{\sApp \mdeepDup #1}
\newcommand{\sLet}[2]{\text{\textsf{let}}\ #1\ \text{\textsf{in}}\ #2}
\newcommand{\sred}[4]{#1 : #2 \Downarrow #3 : #4}
\newcommand{\sRule}[1]{\text{{\textsc{#1}}}}
\newcommand{\fv}[1]{\text{fv}(#1)}
\newcommand{\ufv}[1]{\text{ufv}(#1)}
\newcommand{\ur}[2]{\text{ur}_{#1}(#2)}
\newcommand{\dom}[1]{\text{dom}\,#1}
\newcommand{\fresh}[1]{#1'}
\begin{alignat*}{2}
x,y &\in \sVar
\displaybreak[1]
\\
e &\in
\sExp &&\Coloneqq
\begin{aligned}[t]&
\sLam x e
\mid \sApp e x
\mid x \mid
\\&
\sLet {x_1 = e_1,\ldots,x_n = e_n} e \mid
\\&
\sDup x \mid \sDeepDup x
\end{aligned}
\displaybreak[1]\\
\Gamma, \Delta, \Theta &\in \sHeap &&= \sVar \pfun \sExp
\displaybreak[1]\\
z &\in \sVal &&\Coloneqq \sLam x e
\end{alignat*}
His lambda terms are normalized, i.e.\ all bound variables are distinct and all applications are applications of an expression to a variable.

The set of free variables of an expression $e$ is $\fv e$. Similarly, the set of unguarded free variables $\ufv e$ of an expression $e$, is inductively defined just like $\fv e$ with the exception that $\ufv {\sDeepDup x} = \emptyset$. A value $\hat z$ is $z$ with all bound variables renamed to completely fresh variables.

To avoid having to introduce constructors and case expressions as well we assume $\sDup$ and $\sDeepDup$ to return their result without the wrapping in \li-Box-. This captures the semantics of the Haskell expression\nopagebreak
\begin{haskell}
(\x. let Box y = dup x in y) :: a -> a.
\end{haskell}

In addition to the unmodified reduction rules \sRule{Lam}, \sRule{App}, \sRule{Var} and \sRule{Let}, we add the two rules \sRule{Dup}  and \sRule{Deep} in Figure~\ref{fig:semrules}. The use of $\ufv e$ instead of $\fv e$ in the rule \sRule{DeepDup} is required to avoid a livelock if $\sDeepDup x$ is evaluated while $x$ is itself bound to $\sDeepDup y$.

In the following every heap/term pair $\Gamma : e$ is assumed to be distinctly named, i.e.\ every binding occurring in $\Gamma$ and in $e$ binds a distinct variable; this property is preserved by the reduction rules.

\begin{figure*}
\parskip1cm
\begin{mathpar}
\inferrule
{ }
{\sred{\Gamma}{\sLam xe}{\Gamma}{\sLam xe}}
\sRule{Lam}
\and
\inferrule
{\sred{\Gamma}e{\Delta}{\sLam y e'}\\ \sred{\Delta}{e'[x/y]}{\Theta}{z}}
{\sred\Gamma{\sApp e x}\Theta z}
\sRule{App}
\and
\inferrule
{\sred\Gamma e \Delta z}
{\sred{\Gamma, x\mapsto e} x {\Delta, x\mapsto z}{\hat z}}
\sRule{Var}
\and
\inferrule
{\sred{\Gamma,x_1\mapsto e_1,\ldots,x_n\mapsto e_n} e \Delta z}
{\sred{\Gamma}{\sLet{x_1 = e_1,\ldots, x_n = e_n}e} \Delta z}
\sRule{Let}
\and
\inferrule
{\sred{\Gamma,x\mapsto e, \fresh x\mapsto \hat e} {\fresh x} \Delta z \\ \fresh x \text{ fresh}}
{\sred{\Gamma,x\mapsto e}{\sDup x} \Delta z}
\sRule{Dup}
\and
\inferrule
{
\sred{
\Gamma,
x\mapsto e,
\fresh x\mapsto \hat e[\fresh y_1/y_1,\ldots, \fresh y_n/y_n],
\fresh y_1 \mapsto \sDeepDup y_1,\ldots,
\fresh y_n \mapsto \sDeepDup y_n
} {\fresh x} \Delta z 
\\
\ufv e = \{y_1,\ldots,y_n\}
\\
\fresh x,\ \fresh y_1,\ldots,y_n \text{ fresh}
}
{\sred{\Gamma,x\mapsto e}{\sDeepDup x} \Delta z}
\sRule{Deep}
\end{mathpar}
\caption{Natural semantics extended for \li-dup- and \li-deepDup-}
\label{fig:semrules}
\end{figure*}

\newcommand{\dsem}[2]{\llbracket #1 \rrbracket_{#2}}
\newcommand{\esem}[1]{\{\!\!\{#1\}\!\!\}}
\newcommand{\case}[1]{\par\vspace{\standardvspace}\noindent\textbf{Case:} #1\nopagebreak\par\noindent\ignorespaces}

Besides the natural semantics, Launchbury also defines a denotational semantics. He models values as a lifted function space, denoted \sValue, and environments
\[
\rho \in \sEnv = \sVar \to \sValue
\]
as functions from variables into values. He writes $\rho \le \rho'$ if $\rho'$ extends $\rho$, i.e. they differ only for variables where $\rho$ is bottom. The expression $\dsem e \rho$ is the value of the expression $e$ in the environment $\rho$.

The semantics of a heap $\Gamma$ is given by $\esem \Gamma \rho$, which is the environment $\rho$ updated by the values specified in the heap. This is defined as a fixed point, as the heap may contain recursive references:
\begin{multline*}
\esem{ x_1\mapsto e_1,\ldots,x_n\mapsto e_n}\rho \\
= \mu \rho'. \rho \sqcup (x_1 \mapsto \dsem{e_1}{\rho'}) \sqcup \cdots \sqcup (x_n \mapsto \dsem{e_n}{\rho'})
\end{multline*}
This definition makes sense on environments $\rho$ that are consistent with $\Gamma$, i.e.\ if $\rho$ and $\Gamma$ bind the same variable, then they are bound to values for which an upper bound exists.

Launchbury proves his natural semantics to be correct with respect to the denotational semantics. 
Naturally, we want to preserve this property. Our new primitives should be invisible to the denotational semantics, hence we extend the semantics function as follows:
\begin{align*}
\dsem{\sDup x}\rho &\coloneqq \dsem{x}\rho \\
\dsem{\sDeepDup x}\rho &\coloneqq \dsem{x}\rho.
\end{align*}

\begin{theorem}[Theorem 2 from \citep{launchbury}]
If $\sred\Gamma e \Delta z$, then for all environments $\rho$,
\[
\dsem e {\esem \Gamma \rho} = \dsem z {\esem \Delta \rho}
\text{ and }
\esem \Gamma \rho \le \esem \Delta \rho.
\]
\end{theorem}
\begin{proof}
The proof in \citep{launchbury} is by induction on the derivation; we only have to give it for the two new cases corresponding to the rules \sRule{Dup} and \sRule{Deep}. We assume that the fresh variables in the rules are chosen to be undefined in $\rho$:

\case{$\sDup x$}
By induction, we know (i) $\dsem{\fresh x}{\esem{\Gamma, x \mapsto e, \fresh x \mapsto \hat e} \rho} = \dsem{z}{\esem \Delta \rho}$ and (ii) $\esem{\Gamma, x \mapsto e, \fresh x \mapsto \hat e} \rho \le \esem \Delta \rho$.

For the first part, we have 
\begin{align*}
&\phantom{{}={}}\dsem{\sDup x}{\esem{\Gamma, x\mapsto e}\rho}\\
&= \dsem{x}{\esem{\Gamma, x\mapsto e}\rho}\\
&= \dsem{e}{\esem{\Gamma, x\mapsto e}\rho}\\
&= \dsem{\hat e}{\esem{\Gamma, x\mapsto e}\rho}\\
&= \dsem{\hat e}{\esem{\Gamma, x\mapsto e, \fresh x \mapsto \hat e}\rho} && \text{$\fresh x$ fresh}\\
&= \dsem{\fresh x}{\esem{\Gamma, x\mapsto e, \fresh x \mapsto \hat e}\rho} \\
&= \dsem{z}{\esem \Delta \rho} &&\text{by (i)}
\end{align*}
as desired.

The second part follows from (ii) and from $\fresh x$ being fresh:
\begin{align*}
\esem{\Gamma, x\mapsto e}\rho \le \esem{\Gamma, x\mapsto e, \fresh x\mapsto \hat e} \rho \le \esem{\Delta}\rho
\end{align*}

\case{$\sDeepDup x$}
Let \mbox{$\Gamma'$} denote the heap in the assumption of the rule, i.e. $\Gamma, x\mapsto e, \fresh x\mapsto \hat e[\fresh y_1/y_1,\ldots,\fresh y_n/y_n],\allowbreak \fresh y_1 \mapsto \sDeepDup {y_1},\ldots,\allowbreak\fresh y_n \mapsto \sDeepDup {y_n}$.
By induction, we know (i) $\dsem{\fresh x}{\esem{\Gamma'} \rho} = \dsem{z}{\esem \Delta \rho}$ and (ii) $\esem{\Gamma'} \rho \le \esem \Delta \rho$.

The newly introduced variables $\fresh y_i$, $i=1,\ldots,n$, have the same semantics as their original counterparts:
\[
\dsem{\fresh y_i}{\esem{\Gamma'}\rho}
= \dsem{\sDeepDup {y_i}}{\esem{\Gamma'}\rho}
= \dsem{y_i}{\esem{\Gamma'}\rho}
= \dsem{y_i}{\esem{\Gamma, x\mapsto e}\rho}.
\]
This implies (iii) $\dsem{\hat e[\fresh y_1/y_1,\ldots,\fresh y_n/y_n]}{\esem{\Gamma'}\rho} = \dsem{e}{\esem{\Gamma, x\mapsto e}\rho}$. Hence
{\allowdisplaybreaks[1]
\begin{align*}
&\phantom{{}={}}\dsem{\sDeepDup x}{\esem{\Gamma, x\mapsto e}\rho}\\
&= \dsem{x}{\esem{\Gamma, x\mapsto e}\rho}\\
&= \dsem{e}{\esem{\Gamma, x\mapsto e}\rho}\\
&= \dsem{\hat e[\fresh y_1/y_1,\ldots,\fresh y_n/y_n]}{\esem{\Gamma'}\rho} &&\text{by (iii)}\\
&= \dsem{\fresh x}{\esem{\Gamma'}\rho}\\
&= \dsem{z}{\esem \Delta \rho} &&\text{by (i)}
\end{align*}
}
and, by (ii),
\begin{align*}
\esem{\Gamma, x\mapsto e}\rho \le \esem{\Gamma'} \rho \le \esem{\Delta}\rho.
\end{align*}
\end{proof}

More interesting than the semantic correctness of our additional rules is what properties of \li-deepDup- we can prove with them. Following our intuition from the introduction, we formulate the next theorem, where $\Gamma \subseteq \Delta$ means that $\Gamma$ and $\Delta$ agree on the domain of $\Gamma$ and only new variables are bound.

\begin{theorem}
Consider the expression
\[
e=\sLet{x_1' = \sDeepDup x_1,\ldots,x_n'= \sDeepDup x_n} e'
\]
with $\fv{e'} \subseteq \{x_1',\ldots,x_n'\}$. If $\sred \Gamma e \Delta z$ and $z$ is a closed value (i.e.\ $\fv{z} = \emptyset$), then $\Gamma \subseteq \Delta$.
\label{thm:deepdup}
\end{theorem}
This implies that any value on the heap $\Delta$ that was created during the evaluation of $e$ can be freed afterwards.

The theorem is an immediate consequence of statement (a) of the following Lemma~\ref{lem:deepdup}, with $\Gamma_0 = \Gamma$. We will need the notion of the unguarded reachable set $\ur{\Gamma}{e}$ of an expression $e$ in a context $\Gamma$, which is mutually defined for all expressions as the smallest sets which fulfill the equation
\[
\ur{\Gamma}{e} = \ufv e \cup \textstyle\bigcup_{x\in \ufv e}\ur{\Gamma}{\Gamma\ x}.
\]
Note that $\ufv e \subseteq \ufv {e'}$ implies $\ur\Gamma e \subseteq \ur\Gamma {e'}$. 

\begin{lemma}
\label{lem:deepdup}
Let $\Gamma_0$ be a heap and $U= \dom\Gamma_0$ its domain. If $\sred\Gamma e \Delta z$, $\Gamma_0 \subseteq \Gamma$ and $U \cap \ur \Gamma e = \emptyset$, then
\begin{enumerate}[(a)]
\item $\Gamma_0 \subseteq \Delta$,
\item $U \cap \ur \Delta z=\emptyset$ and
\item \label{lem:deepdup:3} $U \cap \ur\Gamma y = \emptyset$ implies $U\cap \ur\Delta y = \emptyset$ for $y\in \dom\Gamma$.
\end{enumerate}
\end{lemma}

\begin{proof} 
The proof is by induction on the structure of the derivation $\sred\Gamma e \Delta z$.
\case{$\sLam xe$}
Immediate.

\case{$\sApp e x$}
From $\ur\Gamma{\sApp e x} = \ur\Gamma{e} \cup \ur\Gamma{x}$ and the assumption $U \cap \ur\Gamma{\sApp e x}= \emptyset$ we have $U \cap \ur\Gamma e = \emptyset$ and $U \cap \ur\Gamma x = \emptyset$. From the first induction hypothesis we obtain (i) $\Gamma_0 \subseteq \Delta$, (ii) $U \cap \ur\Delta {\sLam y e'}= \emptyset$ and (iii) $U \cap \ur\Delta{x} = \emptyset$.

As $\ur\Delta{e'[x/y]} \subseteq \ur\Delta{\sLam ye'} \cup \ur\Delta{x}$, (ii) and (iii) imply $U \cap \ur\Delta{e'[x/y]} = \emptyset$. With (i) we obtain (a) $\Gamma_0\subseteq\Theta$ and (b) $U\cap \ur\Theta z = \emptyset$ from the second induction hypothesis.

Statement (c) follows immediately from the induction hypothesizes.

\case{$x$}
Removing a variable from a heap does not increase unreachable sets, so $\ur\Gamma e \subseteq \ur{\Gamma, x\mapsto e}e \subseteq \ur{\Gamma,x\mapsto e} x$. From $x \in \ur{\Gamma,x \mapsto e} x$ and the assumption $U \cap \ur{\Gamma,x\mapsto e}x = \emptyset$ we have $x \notin U$, thus $\Gamma_0 \subseteq \Gamma$, and $U \cap \ur\Gamma e = \emptyset$. From the induction hypothesis we now obtain $\Gamma_0 \subseteq \Delta$ and $U \cap \ur\Delta z = \emptyset$. As $\Delta \subseteq (\Delta, x \mapsto z)$, $\ufv{z} = \ufv{\hat z}$ and $\ur \Delta z = \ur {\Delta, x \mapsto z} {\hat z}$, the statements (a) $\Gamma_0\subseteq(\Delta,x\mapsto z)$ and (b) $ U\cap \ur{\Delta, x\mapsto z}{\hat z} = \emptyset$ follow.

Let $y\in \dom\Gamma_0$ with $U\cap \ur{\Gamma,x\mapsto e}y=\emptyset$. As $\ur\Gamma y \subseteq \ur{\Gamma,x\mapsto e} y$ we have $U\cap\ur\Gamma y=\emptyset$ and hence $U\cap\ur{\Delta}y=\emptyset$ from the induction hypothesis. This and (b) imply (c), as $\ur{\Delta,x\mapsto z}y \subseteq \ur{\Delta}y \cup \ur{\Delta}z$.

\case{$\sLet {x_1=e_1,\ldots,x_n=e_n}e$}
For brevity, let $\Gamma' = \Gamma,x_1\mapsto e_1,\ldots,x_n\mapsto e_n$ and $e_{l} = \sLet {x_1=e_1,\ldots,x_n=e_n}e$.
Clearly $\Gamma_0 \subseteq \Gamma \subseteq \Gamma'$.
Also, for each $e_* \in \{e,e_1,\ldots,e_n\}$ we have $\ufv{e_*} \subseteq \ufv{e_l} \cup \{x_1,\ldots,x_n\}$. This implies 
{\allowdisplaybreaks[1]
\begin{align*}
\ur{\Gamma'}e
&= \ufv e \cup \textstyle\bigcup_{x\in \ufv e}\ur{\Gamma'}{\Gamma'\ x} \\
&\subseteq 
\begin{aligned}[t]
\ufv {e_l} &\cup \{x_1,\ldots,x_n\}\\
&\cup \textstyle\bigcup_{x\in \ufv {e_l}}\ur{\Gamma'}{\Gamma'\ x} \\
&\cup \ur{\Gamma'}{\Gamma'\ x_1} \cup \cdots \cup \ur{\Gamma'}{\Gamma'\ x_n}
\end{aligned}\\
&=
\begin{aligned}[t]
\ufv {e_l} &\cup \{x_1,\ldots,x_n\}\\
&\cup \textstyle\bigcup_{x\in \ufv {e_l}}\ur{\Gamma'}{\Gamma'\ x}\\
&\cup \ur{\Gamma'}{e_1} \cup \cdots \cup \ur{\Gamma'}{e_n}
\end{aligned}\\
&= 
\begin{aligned}[t]
\ufv {e_l} &\cup \{x_1,\ldots,x_n\} \\
&\cup \textstyle\bigcup_{x\in \ufv {e_l}}\ur{\Gamma'}{\Gamma'\ x}
\end{aligned}\\
&= \ur{\Gamma'} {e_l} \cup \{x_1,\ldots,x_n\}\\
&= \ur{\Gamma} {e_l} \cup \{x_1,\ldots,x_n\}.
\end{align*}
}
As all bound variables are distinct from variables in the heap, no $x_i\in U$. From $U\cap \ur\Gamma {e_l}
= \emptyset$, we have $U \cap \ur {\Gamma'} e = \emptyset$ and statements (a) and (b) follow from the induction hypothesis.

For $y\in \dom\Gamma$ the unreachable set of $y$ cannot contain any of $x_1,\ldots,x_n$, as the heap/term pair $\Gamma : e_l$ is distinctly named, so we have $\ur\Gamma{y} = \ur{\Gamma'}y$ and (c) follows from the induction hypothesis.

\case{$\sDup x$}
Clearly $\Gamma_0 \subseteq \Gamma, x\mapsto e \subseteq \Gamma, x\mapsto e, \fresh x \mapsto \hat e$. Also,
\begin{align*}
\ur{\Gamma,x\mapsto e, \fresh x\mapsto \hat e}{\fresh x}
&= \ur{\Gamma, x\mapsto e, \fresh x\mapsto \hat e}{\hat e} \cup \{\fresh x\}\\
&= \ur{\Gamma, x\mapsto e}{e} \cup \{\fresh x\}\\
&\subseteq \ur{\Gamma, x\mapsto e}{\sDup x} \cup \{\fresh x\}.
\end{align*}
As $x'$ is fresh, $\fresh x\notin U$ and from $U \cap \ur{\Gamma, x\mapsto e}{\sDup x}= \emptyset$ we have $U \cap \ur{\Gamma,x\mapsto e, \fresh x \mapsto \hat e}{\fresh x}=\emptyset$, so the first statement follows from the induction hypothesis.

Statement (c) follows immediately as $\fresh x$ is fresh.

\case{$\sDeepDup x$}
Let $\Gamma'$ denote the heap $\Gamma, x\mapsto e, \fresh x\mapsto \hat e[\fresh y_1/y_1,\ldots,\fresh y_n/y_n],\allowbreak \fresh y_1 \mapsto \sDeepDup y_1,\allowbreak \ldots,\allowbreak \fresh y_1 \mapsto \sDeepDup y_1$.
Recall that, by definition, $\ufv {\sDeepDup x}=\emptyset$, hence $\ur{\Gamma'}{\sDeepDup x} = \emptyset$. So
{\allowdisplaybreaks[1]
\begin{align*}
\ur{\Gamma'}{\fresh x}
&= \{\fresh x\} \cup \ur{\Gamma'}{\hat e[\fresh y_1/y_1,\ldots,\fresh y_n/y_n]} \\
&= \{\fresh x\}
\begin{aligned}[t]
&\cup \ufv{\hat e[\fresh y_1/y_1,\ldots,\fresh y_n/y_n]}\\
&\cup \textstyle\bigcup_{z \in \ufv{\hat e[\fresh y_1/y_1,\ldots,\fresh y_n/y_n]}} \ur{\Gamma'}{\Gamma' \ z}
\end{aligned}\\
&= \{\fresh x\}
\begin{aligned}[t]
&\cup \{\fresh y_1,\ldots, \fresh y_n\}\\
&\cup \textstyle\bigcup_{i=1,\ldots,n} \ur{\Gamma'}{\Gamma'\ \fresh y_i}
\end{aligned}\\
&= \{\fresh x\} 
\begin{aligned}[t]
&\cup \{\fresh y_1,\ldots, \fresh y_n\}\\
&\cup \textstyle\bigcup_{i=1,\ldots,n} \ur{\Gamma'}{\sDeepDup {y_i}}
\end{aligned}\\
&= \{\fresh x\} \cup \{\fresh y_1,\ldots, \fresh y_n\}
\end{align*}
}
and, as these are all fresh variables, $U \cap \ur{\Gamma'}{\fresh x} = \emptyset$. Clearly, $\Gamma_0 \subseteq (\Gamma, x\mapsto e) \subseteq \Gamma'$, so the  first statement follows from the induction hypothesis.

Statement (c) follows immediately as the additional variables are fresh.
\end{proof}

Having cast our intuition of \li-dup- and \li-deepDup- into a precise form using a formal semantics, we now explain how we have implemented this semantics, or rather a pragmatic approximation, in a real environment.

\section{The prototype implementation}
\label{sec:prototype}

Our implementation\footnote{Available at \url{http://darcs.nomeata.de/ghc-dup}} works with the Glasgow Haskell Compiler (GHC), version 7.4.1, and requires no modifications to the compiler or its runtime: The code is compiled to a usual object file, linked into the resulting binary and called via the foreign function interface.

GHC compiles Haskell code first to a polymorphic, explicitly typed lambda-calculus called \emph{Core} \citep{core,system-fc}, then to the \emph{Spineless Tagless G-machine} (STG) \citep{stg}. From there, it generates \emph{Cmm} code, an implementation of the portable assembly language C-{}- which is then compiled to machine code, either directly or via LLVM.

Our work looks at objects in the sense of the STG, so we only need to worry about data representation on the heap \citep{stg}. Design decisions regarding the earlier transformations, such as the evaluation model \cite{evalapply}, are thus not important here.

\def\ux{2.2cm}\def\uy{0.6cm}
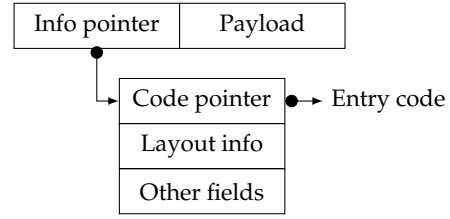
\begin{figure}
\begin{center}
\begin{tikzpicture}[x=\ux, y=\uy,word/.style={shape=rectangle, draw, minimum width=\ux, minimum height=\uy},>=latex]
\draw (0,0) rectangle +(1,1) node[midway] (ip) {Info pointer};
\draw (1,0) rectangle +(1,1) node[midway] {Payload};
  (0,0) node[word] (ip) {Info pointer}
++(0,-1) node[word, minimum width=2*\ux] {Payload};

\begin{scope}[yshift=-0.7cm, xshift=2.5cm]
\draw
  (0,0) node[word] (tbl) {Code pointer}
++(0,-1) node[word] {Layout info}
++(0,-1) node[word] {Other fields};
\end{scope}
\draw[*->] (ip.south) |- (tbl.west);
\draw[*->] (tbl.east) -- ++(.5cm,0) node[right] {Entry code};
\end{tikzpicture}
\end{center}
\caption{The common layout of heap objects}
\label{fig:heap}
\end{figure}

The common layout of all objects, or closures,  on the heap is a pointer to a statically allocated \emph{info table}, followed by the payload (Figure \ref{fig:heap}). The info table indicates the type of the object (not to be confused with the type from the type system – these are completely irrelevant at this stage), contains layout information about the payload required by the garbage collector, namely what words are pointers to other objects and what words are not, and the code to be run when the object is evaluated.

There are various types of objects on the heap, most important are:
\begin{itemize}
\item \emph{Data constructors}, representing fully evaluated values. The payload are pointers to the parameters of the constructor.
\item \emph{Function closures}, representing functions. Locally defined functions capture their free variables, these are stored in the payload.
\item \emph{Thunks}, which are unevaluated expressions. Again, the payload contains references to their free variables.
\item \emph{Applications} of a function to a number of arguments. This closure type is usually only used by the GHC interpreter, but we use it in the implementation of \li-deepDup-.
\item \emph{Indirections}, which point to another object on the heap in their payload. These are created during evaluation and removed by the garbage collector.
\end{itemize}

When a thunk is evaluated, it is replaced by an indirection which points to the result of the evaluation, which can be a data constructor or a function closure. This way, when another reference to the thunk is evaluated, the computation is not repeated but the calculated result is used directly, hence the result is \emph{shared}. The indirections do not stay around forever: The next garbage collector run, which copies all live data, will replace references to indirections by whatever the indirection points to.

As we want to avoid this sharing, we need to prevent the original reference to be replaced by the indirection. We cannot change the code of the thunk, but we can copy the thunk, thus creating a new copy that is not referenced by other code, and then evaluate that.
The essence of the surprisingly simple code is listed in Figure \ref{fig:dupcode}; the closure to duplicate is passed in the register \ci-R1- and \ci-Hp- is the heap pointer which is increased by \ci-ALLOC_PRIM-.

\begin{figure}
\begin{cmm}
dupClosure {
    clos = UNTAG(R1);
    // Allocate space for the new closure
    (len) = foreign "C" closure_sizeW(clos "ptr") [];
    ALLOC_PRIM(WDS(len), R1_PTR, dupClosure);
    copy = Hp - WDS(len) + WDS(1);
    p = 0;
    for: // Copy the info pointer and payload
    if(p < len) {
        W_[copy + WDS(p)] = W_[clos + WDS(p)];
        p = p + 1;
        goto for;
    }
    RET_P(copy);
}
\end{cmm}
\caption{The Cmm code for \li-dup-}
\label{fig:dupcode}
\end{figure}

As discussed in Section \ref{sec:deepdup}, this simple approach is not always sufficient, and we want a recursive variant, \li-deepDup-. This function, shown in Figure \ref{fig:deepdupcode}, needs to access the info table of the closure to figure out what part of the payload is a pointer to another heap object. For every referenced object, an application thunk is created which applies \li-deepDup- (or rather the variant \li!deepDupFun! with the better suited type \li!a -> a!), unless we are about to \li-deepDup- a \li-deepDup- thunk. In that case, we just copy it, but leave the argument alone, reflecting the use of $\ufv e$ instead of $\fv e$ in the Rule \sRule{Deep} in the formal semantics. The code listing does not include a few shortcuts, e.g.\ data constructors without pointer arguments such as integer values are not copied.

\begin{figure}
\begin{cmm}
deepDupClosure {
    clos = UNTAG(R1);
    // Allocate space for the new closure
    (len) = foreign "C" closure_sizeW(clos "ptr") [];
    ptrs  = TO_W_(%INFO_PTRS(%GET_STD_INFO(clos)));
    bytes = WDS(len) + ptrs * SIZEOF_StgAP + WDS(ptrs);
    ALLOC_PRIM(bytes, R1_PTR, dupClosure);
    copy = Hp - WDS(len) + WDS(1);
    p = 0;
    for1: // Copy the info pointer and payload
    if(p < len) {
        W_[copy + WDS(p)] = W_[clos + WDS(p)];
        p = p + 1;
	goto for1;
    }
    // Do not wrap \textup{deepDup} thunks again
    if (W_[copy] == stg_ap_2_upd_info &&
        W_[copy + WDS(1)] == Dup_deepDupFun_closure) {
       goto done;
    }
    if 
    p = 0;
    for2: // Wrap all referenced closures in \textup{deepDup} thunks
    if(p < ptrs) {
	ap = Hp - bytes + WDS(1)
	     + p * SIZEOF_StgAP + WDS(p);
        W_[ap] = stg_ap_2_upd_info;
        W_[ap + WDS(1)] = Dup_deepDupFun_closure;
	W_[ap + WDS(2)] = W_[clos + WDS(p)];
	W_[copy + WDS(p)] = ap;
	p = p + 1;
	goto for2;
    }
    done:
    RET_P(copy);
}
\end{cmm}
\caption{The Cmm code for \li-deepDup-}
\label{fig:deepdupcode}
\end{figure}

\subsection{Limitations of the implementation}
\label{sec:shortcomings}

Our implementation is but a prototype; it does not yet work in all situations. One large problem is posed by statically allocated thunks: A value, say \li-nats = [0..]-, defined at the module level is compiled to a thunk with closure type \ci-THUNK_STATIC-, also called a constant applicative form (CAF), and receives special treatment by the garbage collector. Copying such a closure to the heap using the code above would make the garbage collector abort, as it does not expect a static thunk to be found on the heap. But it is not possible to change the type of the closure, as the info table containing the type lies directly next to the code. And in order to create a modified info table somewhere else, the code needs to be copied as well. Therefore, \li-dup- and \li-deepDup- currently does not work for static thunks. When it is passed such a thunk, it prints a warning and returns the original reference, retaining sharing.

It should be possible for \li-dup- to support static thunks with some additional information in the compiled code.  Currently, when execution enters a static thunk and the stack and heap checks have been passed, the thunk is replaced by an indirection into the heap and an update frame is pushed on the stack. If there was a way to jump over the code that sets up the indirection and update frame, e.g.\ via an alternative entry point included in the info table, \li-dup- could create a thunk on the heap that calls the static thunk via this route, effectively kicking off evaluation without affecting the original static thunk.  For \li-deepDup- things are more complicated, as references to static objects are not part of the heap object, but are scattered throughout the machine code. Moving these references to the heap would solve the issue here at hand, but is clearly too expensive.

Also, the prototype does not take multithreaded programs into account and will likely produce bad results when used in such an environment, e.g.\ when another thread replaces a thunk by an indirection during the thunk copy loop in \ci-dupClosure-. Similarly, there are several specialized closure type (arrays, mutable references, weak pointers\citep{weakpointers} and others \citep[page HeapObjects]{commentary}). For each of them, we would need to determine whether they can be safely duplicated and if so, whether this is actually useful.

In the presence of Lazy IO, duplicating thunks can be outright dangerous: Not only can the original and the duplicated thunk evaluate to different values but this can make the program crash, e.g.\ when one copy is done evaluating and causes a file to be closed, while the second copy continues to read from it. Generally everything implemented with \li-unsafePerformIO- is prone to behave badly when combined with \li-dup- or \li-deepDup-.

Function closures need special treatment as there are cases where code assumes a certain reference to always be a function closure and never a thunk that will evaluate to a function. But this is what \li-deepDup- wants to create. Currently, \li-deepDup- will in this case leave the reference as it is. A solution would be to copy the function closure eagerly, so that the reference in the copy again points to a function closure. This would require more sophisticated code to detect cycles. 

\section{Conclusions and further work}

While Haskell gives the programmer great devices to get their programs to do the right thing, such as referential transparency and the type system, she has less means to analyze and control their runtime behavior. Several commercial users have mentioned this as one of the main drawbacks of Haskell \citep{sampson,wehr,hesselink}. This problem deserves more attention and we hope that this work is one step towards a Haskell with better controllable and understandable time and space behavior.

We have shown the feasibility of an explicit sharing-preventing operator in a lazy functional language. We provided two variants, \li-dup- and \li-deepDup-, the former is simpler, but possibly more subtle to put to use effectively, the latter works more predictably, but may impose a larger performance penalty. This is, on a prototypical level, possible with an unmodified Haskell compiler.

As described in Section \ref{sec:shortcomings}, there is work to be done on the implementation before it can be used in production code. Some of that might require changes to the compiler code. Given how sensitive the code is to changes in the runtime representation of Haskell values, a productive version of \li-dup- would probably have to be shipped along with the compiler.


\acks

I would like to thank Andreas Lochbihler for fruitful discussions and proof-reading and the anonymous referees for being supportive of the idea and constructive about the presentation. This work was supported by the Deutsche Telekom Stiftung.

\bibliographystyle{abbrvnat}
\bibliography{bib}

\end{document}